\documentclass[
prl,
twocolumn,
showpacs,
superscriptaddress,
longbibliography]{revtex4-2}

\usepackage{amsfonts}
\usepackage{amsmath}
\usepackage{amssymb}
\usepackage{amsthm}
\usepackage{bm}
\usepackage{dcolumn}
\usepackage{graphicx}
\usepackage{graphics}
\usepackage[latin1]{inputenc}
\usepackage{latexsym}
\usepackage{rotating}
\usepackage{hyperref}
\usepackage[all]{hypcap} 
\usepackage{xspace} 
\usepackage[usenames,dvipsnames]{color}
\usepackage{mathrsfs}
\usepackage{subfigure}
\usepackage{aas_macros}
\usepackage{booktabs}
\usepackage{siunitx}
\usepackage{appendix}
\usepackage{tensor}

\definecolor{mred}{RGB}{127,0,25}
\definecolor{mdgr}{RGB}{51,51,51}
\definecolor{mag}{RGB}{211, 54, 130}
\definecolor{verm}{RGB}{164, 25, 0}

\hypersetup{colorlinks=true,citecolor=NavyBlue,linkcolor=NavyBlue,urlcolor=NavyBlue}

\newcommand{\D}{\nabla}

\newcommand{\pd}{{\partial}}

\newcommand{\svctv}[3]{\tensor[_{(#2)}]{#1}{^{#3}}}
\newcommand{\svcov}[3]{\tensor[^{(#2)}]{#1}{_{#3}}}

\newtheorem{lemma}{Lemma}

\begin{document}


\title{Square Peg in a Circular Hole: \\ Choosing the Right Ansatz for Isolated Black Holes in Generic Gravitational Theories}

\author{Yiqi Xie}
\email{yiqixie2@illinois.edu}
\affiliation{Illinois  Center  for  Advanced  Studies  of  the  Universe and Department of Physics, University of Illinois at Urbana-Champaign, Urbana, Illinois 61801, USA}

\author{Jun Zhang}
\email{jun.zhang@imperial.ac.uk}
\affiliation{Theoretical Physics, Blackett Laboratory, Imperial College, London, SW7 2AZ, United Kingdom}
\affiliation{Illinois  Center  for  Advanced  Studies  of  the  Universe and Department of Physics, University of Illinois at Urbana-Champaign, Urbana, Illinois 61801, USA}

\author{Hector O. Silva}
\email{hector.silva@aei.mpg.de}
\affiliation{Max-Planck-Institut f\"ur Gravitationsphysik (Albert-Einstein-Institut), Am M\"uhlenberg 1, D-14476 Potsdam, Germany}
\affiliation{Illinois  Center  for  Advanced  Studies  of  the  Universe and Department of Physics, University of Illinois at Urbana-Champaign, Urbana, Illinois 61801, USA}

\author{Claudia de Rham}
\email{c.de-rham@imperial.ac.uk}
\affiliation{Theoretical Physics, Blackett Laboratory, Imperial College, London, SW7 2AZ, United Kingdom}
\affiliation{CERCA, Department of Physics, Case Western Reserve University, 10900 Euclid Avenue, Cleveland, Ohio 44106, USA}

\author{Helvi Witek}
\email{hwitek@illinois.edu}
\affiliation{Illinois  Center  for  Advanced  Studies  of  the  Universe and Department of Physics, University of Illinois at Urbana-Champaign, Urbana, Illinois 61801, USA}

\author{Nicol\'as Yunes}
\email{nyunes@illinois.edu}
\affiliation{Illinois  Center  for  Advanced  Studies  of  the  Universe and Department of Physics, University of Illinois at Urbana-Champaign, Urbana, Illinois 61801, USA}

\date{\today}

\begin{abstract}
The metric of a spacetime can be greatly simplified if the spacetime is circular. 
We prove that in generic effective theories of gravity, the spacetime of a stationary, axisymmetric and asymptotically flat solution must be circular if the solution can be obtained perturbatively from a solution in the general relativity limit. This result applies to a broad class of gravitational theories that include arbitrary scalars and vectors in their light sector, so long as their nonstandard kinetic terms and nonmininal couplings to gravity are treated perturbatively.
\end{abstract}


\maketitle

\emph{Introduction.---}%
Despite the complexity and nonlinearity of the Einstein equations, rotating black holes in general relativity (GR) are described by a remarkably simple analytical solution obtained by Kerr~\cite{Kerr:1963ud,Kerr:2009GReGr}. A crucial 
step in finding the Kerr solution is that the ten unknown functions of four coordinate variables in the metric can be reduced to four unknown functions of only two variables. This simplification is only possible because stationary and axisymmetric vacuum solutions in GR belong to a specific class called ``circular spacetimes''~\cite{Papapetrou:1966zz}.
%
%
%
%
However, this is not necessarily the case in generic gravitational theories~\cite{Berti:2015itd}
and one should not expect \textit{a priori} that  black hole solutions in such theories will be circular.
%
%
%
%
In particular, one should expect 
the validity of the circularity assumption to play a role as important as it did in GR
to obtain rotating black hole solutions (either numerically or analytically) in such theories.
%
In turn, knowledge of these solutions constitutes the stepping stone on which many tests of strong-field gravity rely~\cite{Yagi:2016jml}. 
%
The use of an oversimplified ansatz based on the circularity condition can lead to spacetimes that are 
inconsistent with a given theory's field equations. 
This was recently observed, for instance, in the case of rotating black hole solutions with linearly time-dependent hair in cubic Galileon theories in which the circularity condition is not satisfied~\cite{VanAelst:2019kku}.
%
%

In this Letter, we investigate the circularity of stationary and axisymmetric solutions in generic gravitational theories, paving the way for finding rotating black hole solutions in GR and beyond. In order to remain generic on the gravitational theory, we work within the effective field theory (EFT) framework in which UV modifications of GR manifest as higher dimensional operators in the low-energy EFT and can be treated perturbatively.
The EFT framework works well for isolated astrophysical black holes, which have masses in the 
$\sim 5$--$10^{10} M_\odot$ range thanks to their low energy scale ($\lesssim 10^{-11}~{\rm eV}$).
The framework is also supported by the agreement of GR predictions 
%
%
with gravitational wave detections~\cite{LIGOScientific:2019fpa} and other electromagnetic observations~\cite{Abuter:2020dou}.
In particular, we focus on gravitational theories whose low-energy EFT represents extensions of GR involving additional (scalar) fields and other operators.
These EFTs include $f(R)$ gravity or more general scalar-tensor theories~\cite{Sotiriou:2008rp,Kobayashi:2019hrl} and quadratic gravity~\cite{Yagi:2015oca} such as dynamical Chern-Simons gravity~\cite{Jackiw:2003pm,Alexander:2009tp}
and Einstein-dilaton-Gauss-Bonnet gravity~\cite{Metsaev:1986yb,Kanti:1995vq}, as well as gravitational EFTs without light scalar fields, like those studied in~\cite{Endlich:2017tqa,Sennett:2019bpc,deRham:2020ejn}.

As the modifications of GR are small, black hole solutions in the EFTs can be obtained through a perturbative expansion around one (or more) coupling constants of such theories (see~\cite{Campbell:1990fu,Campbell:1990ai,Campbell:1991kz,Campbell:1992hc,Mignemi:1992nt,Yunes:2011we,Pani:2011gy,Yagi:2012ya,Ayzenberg:2014aka,Maselli:2015tta,Maselli:2015yva,Maselli:2017kic,Cardoso:2018ptl,Julie:2019sab,Cano:2019ore} for examples).
We show here that the spacetime of stationary, axisymmetric, and asymptotically flat solutions is circular in these EFTs, hence also in the corresponding high-energy gravitational theories.
In principle, there could be other branches of solutions that are not connected perturbatively to their GR counterparts (see~\cite{Doneva:2017bvd,Silva:2017uqg} for example), but these are not the focus of this Letter.
We use geometric units ($c=8 \pi G=1$) and employ the $(-,+,+,+)$ metric signature.

\emph{Circular spacetimes in GR.---}%
Consider a stationary and axisymmetric spacetime associated with two Killing vectors $\xi^\mu$ and $\chi^\mu$ that correspond to the two isometries, respectively.
Figure~\ref{fig:geom} gives a schematic illustration of this geometry.
%
Carter~\cite{Carter:1970ea} showed that the two Killing vectors commute, which means one can choose adapted coordinates $(t,r,\theta,\phi)$ on the spacetime such that $\xi = \pd_t$ and $\chi= \pd_\phi$. The isometries imply
\begin{eqnarray}\label{eq:ignorable}
\pd_t g_{\mu\nu}= 0 = \pd_\phi g_{\mu\nu} \,.
\end{eqnarray}
Moreover, there exist privileged 2-dimensional surfaces, called ``surfaces of transitivity,'' to which the Killing vectors are everywhere tangent except on the rotation axis where $\chi^\mu$ vanishes. In adapted coordinates, the surfaces of transitivity can be labeled by the values of $(r, \theta)$.
%
\begin{figure}[tbp]
    \centering
    \includegraphics[width=3.375in]{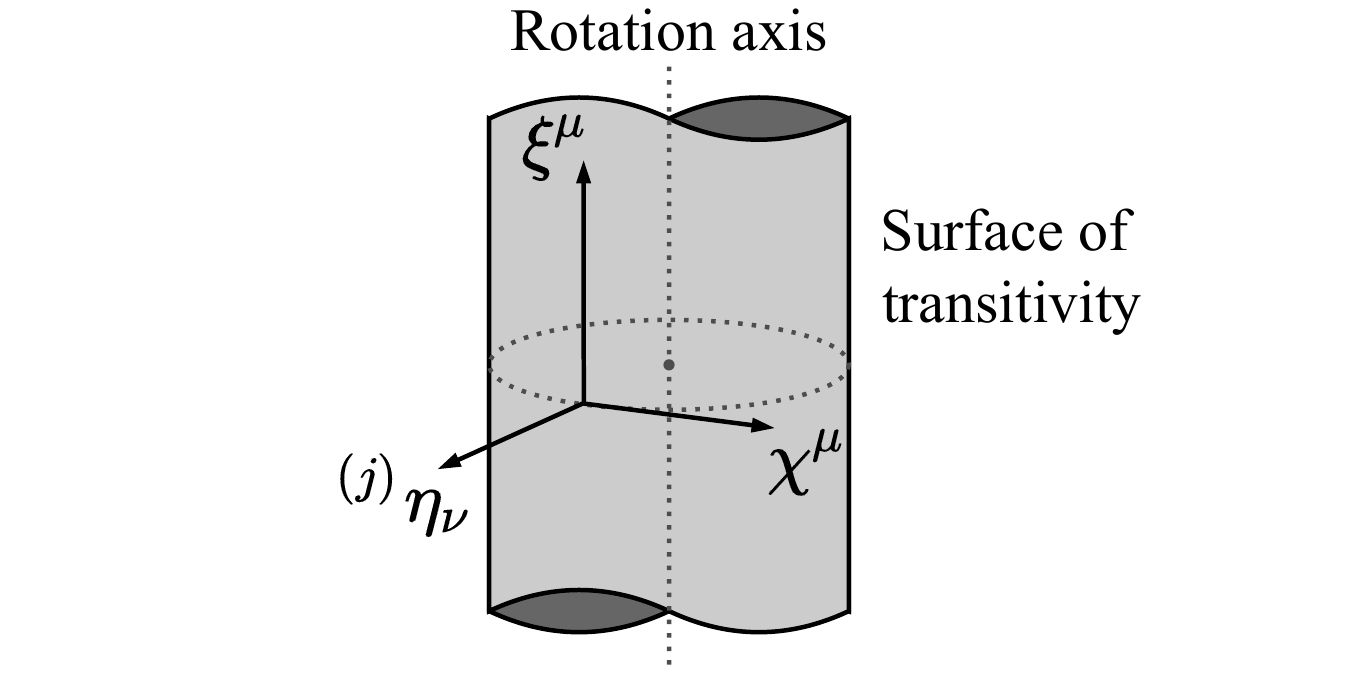}
    \caption{Geometry of a stationary and axisymmetric spacetime. The Killing vector $\xi^\mu$ is associated with time translation and $\chi^\mu$ is associated with rotations about the symmetry axis. Note that $\xi^\mu$ and $\chi^\mu$ are not necessarily orthogonal. The surface of transitivity is generated by $\xi^\mu$ and $\chi^\mu$ and is degenerate on the rotation axis where $\chi^\mu$ vanishes. The independent vectors ${}^{(j)}\eta_\nu~(j=3,4)$ are chosen to be orthogonal to the surface of transitivity. Here we only show one of the orthogonal vectors.}
    \label{fig:geom}
\end{figure}

A circular spacetime is a subclass of stationary and axisymmetric spacetimes for which, in addition to Eq.~\eqref{eq:ignorable}, there exists a family of 2-surfaces known as meridional surfaces that are everywhere orthogonal to the surfaces of transitivity. 
In this case, one can further choose the coordinates $r$ and $\theta$ such that
\begin{eqnarray}\label{eq:goff}
g_{tr}=g_{t\theta}=g_{\phi r}= g_{\phi \theta}=0 \,.
\end{eqnarray}
Without loss of generality, the metric can then take the following ansatz
\begin{align}\label{eq:metric}
g_{\mu\nu} dx^\mu dx^\nu &= - N^2 dt^2 + A^2\left( dr^2+ r^2 d\theta^2\right)
\nonumber \\
&\quad+ B^2r^2 \sin^2\theta \left(d\phi - \omega dt \right)^2 \,,
\end{align}
in ``quasi-isotropic coordinates,'' where $N$, $A$, $B$, and $\omega$ are functions of $r$ and $\theta$.

Papapetrou~\cite{Papapetrou:1966zz} (see also \cite{Wald:1984rg}) showed that a spacetime is circular if (i) $\xi_{[\mu} \chi_\nu \D_\rho \xi_{\sigma]}$ and $\xi_{[\mu} \chi_\nu \D_\rho \chi_{\sigma]}$ each vanish at least at one point of the spacetime, and (ii)
\begin{eqnarray}\label{eq:c-condition}
\xi^\mu \tensor{R}{_\mu^{[\nu}}\xi^\rho \chi^{\sigma]} =0, \quad \chi^\mu \tensor{R}{_\mu^{[\nu}}\xi^\rho \chi^{\sigma]} =0\,,
\end{eqnarray}
everywhere in spacetime, where the square brackets denote full antisymmetrization.
For asymptotically flat spacetimes, which we focus on, Carter further showed a rotation axis at which $\chi^\mu = 0$ exists~\cite{Carter:1970ea}; thus, the first condition is satisfied.

The Eq.~\eqref{eq:c-condition} condition is trivially satisfied if the Ricci tensor vanishes, which means that any stationary, axisymmetric and asymptotically flat vacuum solution in GR is circular, as well as those Ricci-flat solutions in modified gravity theories (e.g.~\cite{Yunes:2011we,Motohashi:2018wdq}). 
%
The Eq.~\eqref{eq:c-condition} condition can also be recast as a requirement of the Ricci tensor being
``invertible''~\cite{Carter:1969zz,Carter:1973rla,carter2009republication}.
%
%
Let $\svctv{\zeta}{i}{\mu}\,(i=1,2)$ be the two Killing vectors $\xi^\mu$ and $\chi^\mu$, and $\svcov{\eta}{j}{\nu}\,(j=3,4)$ be two independent vectors everywhere orthogonal to $\xi^\mu$ and $\chi^\mu$. 
A tensor is said to be invertible in the isometry group if the scalars obtained by contracting any combinations of the tensor's indices with any choice of $\svctv{\zeta}{i}{\mu}$ and $\svcov{\eta}{j}{\nu}$ vanish whenever the number of contracted $\svctv{\zeta}{i}{\mu}$ is odd.
In particular, the Ricci tensor is invertible if
\begin{eqnarray}\label{eq:c-condition_inv}
\tensor{R}{_\mu^\nu}\, \svctv{\zeta}{i}{\mu}\, \svcov{\eta}{j}{\nu} = 0\,,
\quad i=1,2\,,\quad j=3,4\,.
\end{eqnarray}
Heuristically, the Eq.~\eqref{eq:c-condition_inv} condition is equivalent to the Eq.~\eqref{eq:c-condition} condition because the latter is equivalent to requiring that $ \svctv{\zeta}{i}{\mu}\, \, \tensor{R}{_\mu^\nu}\,$ be tangent to the surface of transitivity (i.e., proportional to any linear combination of $\svctv{\zeta}{i}{\nu}$), and thus, that any part tangent to the meridional surface [i.e., proportional to any linear combination of $\svctv{\eta}{j}{\nu}$] vanish.
In the following, we shall omit the presub and superscript of $\zeta^\mu$ and $\eta_\nu$ and bear in mind that each of them represents a vector in a two vector set.

\emph{Circularity in generic gravitational theories.---}%
Let us consider a generic gravitational theory, potentially containing fields of arbitrary spin and coupling to gravity with the Lagrangian
\begin{eqnarray}\label{eq:Lgen}
{\cal L} = \frac{1}{2}R + {\cal L}_{\varphi}  + {\cal L}_{\psi} + {\cal L}_{\rm int}\left(\nabla_\rho,\, R_{\rho\sigma\alpha\beta},\, \varphi,\, \psi \right),
\end{eqnarray}
where the fields are classified as heavy fields $\psi$ or light fields $\varphi$ depending on whether their masses are above or below the curvature scale of the solution that we are interested in. Here, ${\cal L}_{\varphi}$ and ${\cal L}_{\psi}$ are the Lagrangians of $\varphi$ and $\psi$, while ${\cal L}_{\rm int}$ captures all the interactions between the fields as well as any nonminimal couplings to gravity. In particular, we assume that nonstandard kinetic terms of $\varphi$, if there is any in ${\cal L}_{\varphi}$, can be treated perturbatively.
At the energy scale of the solution, we can integrate out the heavy fields with masses larger than the curvature of the solution we are interested in,
\begin{eqnarray}
e^{i \int d^4x \sqrt{-g}\, {\cal L}_{\rm EFT}} = \int {\cal D} \psi \, e^{i \int d^4x \sqrt{-g}\, {\cal L}},
\end{eqnarray}
and obtain a low-energy EFT with the following Lagrangian (see Refs.~\cite{Avramidi:1990je,Avramidi:1986mj,deRham:2019ctd} for explicit examples):
\begin{eqnarray}\label{eq:Ltot}
{\cal L}_{\rm EFT} = \frac{1}{2}R + {\cal L}_0 \left(\varphi, \, g_{\mu\nu} \right) + \alpha\, {\cal L}_{\rm M} \left(\nabla_\rho,\, R_{\rho\sigma\alpha\beta},\, \varphi\right)\,,\
\end{eqnarray}
where operators are sorted according to their dimensions. In particular, ${\cal L}_0$ are operators constructed by the light fields $\varphi$ and their covariant derivatives with dimensions equal to or less than 4, while ${\cal L}_{\rm M}$ are higher dimension operators constructed by the Riemann tensor, the light fields, and derivatives of both, which therefore are suppressed by a small parameter $\alpha$. The heavy fields $\psi$ in Eq.~\eqref{eq:Lgen} have been integrated out and manifest themselves solely as higher curvatures and derivative corrections in ${\cal L}_{\rm M}$. The curvature scale of isolated astrophysical black holes is expected to be smaller than $10^{-11}$~eV. Hence, in realistic situations, the heavy fields $\psi$ include all massive particles of the standard model and beyond.
 
For now, we focus on the case in which the light fields, if any, are all scalar fields. We emphasize that $\varphi$ denotes all light fields in the EFT, which we shall not distinguish with additional labels, and thus, inner products require an internal space metric, which we will also suppress~\cite{Damour:1992we}. This EFT reduces identically and smoothly to GR as $\alpha \to 0$, i.e.,~in this limit Eq.~\eqref{eq:Ltot} reduces to the Einstein-Hilbert action minimally coupled to light scalar fields.

The modified Einstein equations in this theory are
\begin{eqnarray}\label{eq:EoM}
R_{\mu\nu}-\frac{1}{2} R\, g_{\mu\nu} = T_{\mu\nu} + \alpha M_{\mu\nu},
\end{eqnarray} 
where $T_{\mu\nu} \equiv - 2\delta \left(\sqrt{-g}{\cal L}_0\right)/\delta g^{\mu\nu}$ and $M_{\mu\nu} \equiv - 2\delta \left(\sqrt{-g}{\cal L}_{\rm M}\right)/\delta g^{\mu\nu}$ are the energy-momentum tensors associated with ${\cal L}_0$ and ${\cal L}_{\rm M}$, respectively. 
In particular, terms in $T_{\mu\nu}$ are either proportional to $g_{\mu\nu}$ or proportional to $\pd_\mu \varphi \pd_\nu \varphi$ due to the dimension of the operators in ${\cal L}_0$. Given the smallness of $\alpha$, a solution to Eq.~\eqref{eq:EoM} $\{g_{\mu\nu}, \varphi\}$ can be obtained order by order in $\alpha$. For concreteness, we use $\{g_{\mu\nu}^{(n)}, \varphi^{(n)}\}$ to denote the solution to the $n$th order in $\alpha$, i.e.,~$g_{\mu\nu} = g_{\mu\nu}^{(n)} +{\cal O}(\alpha^{n+1})$,
with ${\cal O}(\alpha^{n+1})$ accounting for all higher-order corrections.
We also label a quantity with subscript or superscript $(n)$, e.g.,~$T^{(n)}_{\mu\nu}$ or $g_{(n)}^{\mu\nu}$, if it is calculated up to the $n$th order in $\alpha$. The full solution is given by $\{g_{\mu\nu}^{(n)}, \varphi^{(n)}\}$ with $n$ approaching infinity for a sufficiently small $\alpha$.

In the following, we prove that the spacetime of a stationary, axisymmetric, and asymptotically flat solution is necessarily circular if the solution can be obtained order by order in $\alpha$.
Here we only consider solutions with stationary and axisymmetric scalar fields, which are not necessarily required for the spacetime also be stationary and axisymmetric as we 
discuss later. The proof can be done in three steps.

First, we prove that the solution is circular at zeroth order in $\alpha$, i.e.,~$g_{\mu\nu}^{(0)}$ is circular. At zeroth order, we get back to GR, and $g_{\mu\nu}^{(0)}$ is circular if $T_{\mu\nu}^{(0)}$ is invertible \cite{Carter:1969zz}. In order to show the invertibility, let us consider $\tensor{T}{_\mu^\nu} \tensor{\zeta}{^\mu}\, \tensor{\eta}{_\nu}$, where $\zeta^\mu$ are the two Killing vectors and $\eta_\nu$ are the two independent vectors orthogonal to $\zeta^\mu$. Since the scalar fields are stationary and axisymmetric, the vanishing of their Lie derivatives along $\zeta^\mu$ implies
\begin{eqnarray}\label{eq:lie_dphi}
{\pounds}_{\zeta} \varphi \equiv \zeta^{\mu} \pd_\mu \varphi = 0.
\end{eqnarray}
Thus, terms in $T_{\mu\nu}$ that are proportional to $\pd_\mu \varphi \pd_\nu \varphi$ vanish after contracting with $\zeta^\mu$. The rest of $T_{\mu\nu}$ is proportional to $g_{\mu\nu}$ and do not contribute to $\tensor{T}{_\mu^\nu} \tensor{\zeta}{^\mu}\, \tensor{\eta}{_\nu}$ given the orthogonality between $\zeta^\mu$ and $\tensor{\eta}{_\nu}$. Therefore,
\begin{eqnarray}\label{eq:T}
\tensor{T}{_\mu^\nu} \tensor{\zeta}{^\mu}\, \tensor{\eta}{_\nu}=0,
\end{eqnarray}
i.e.,~$\tensor{T}{_\mu_\nu}$ is invertible. At zeroth order in $\alpha$, Eq.~\eqref{eq:T} means $T_{\mu\nu}^{(0)}$ is invertible, and hence $g_{\mu\nu}^{(0)}$ is circular.

Next, we prove that if $g_{\mu\nu}^{(0)}$ is circular, then $g_{\mu\nu}^{(1)}$ is also circular. This can be proved if the Ricci tensor associated with $g_{\mu\nu}^{(1)}$ is invertible, or equivalently
\footnote{To be precise, Eq.~\eqref{eq:R1} does not immediately indicate $g_{\mu\nu}^{(1)}$ is circular. Instead, it indicates the metric corresponding to $\left(\tensor{R}{_\mu^\nu}\right)^{(1)}$, which could be different from $g_{\mu\nu}^{(1)}$ at ${\cal O}(\alpha^2)$, is circular. Nevertheless, it means there exist coordinates, in which certain components of the metric corresponding to $\left(\tensor{R}{_\mu^\nu}\right)^{(1)}$ vanish (cf. Eq.~\eqref{eq:goff}). Those components remain zero after truncating all ${\cal O}(\alpha^2)$ terms, in which case the metric corresponding to $\left(\tensor{R}{_\mu^\nu}\right)^{(1)}$ reduces to $g_{\mu\nu}^{(1)}$. Thus, the Ricci tensor associated with $g_{\mu\nu}^{(1)}$ is invertible and $g_{\mu\nu}^{(1)}$ is circular.},
\begin{eqnarray}\label{eq:R1}
\tensor{R}{_\mu^\nu} \zeta^\mu \,\tensor{\eta}{_\nu}= 0 + {\cal O}(\alpha^2).
\end{eqnarray}
Contracting Eq.~\eqref{eq:EoM} with $\zeta^\mu$ and $\eta^{\nu}$, we find
\begin{eqnarray}\label{eq:M}
\tensor{R}{_\mu^\nu} \zeta^\mu \,\tensor{\eta}{_\nu} = \alpha \tensor{M}{_\mu^\nu} \zeta^\mu\, \tensor{\eta}{_\nu},
\end{eqnarray}
where the second term on the left hand side of Eq.~\eqref{eq:EoM} does not contribute due to the orthogonality between $\zeta^\mu$ and $\eta_\nu$, and the first term on the right hand side of Eq.~\eqref{eq:EoM} also vanishes because of the invertibility of $T_{\mu\nu}$. 

\newcounter{aux_footnote}
On the other hand, since $g_{\mu\nu}^{(0)}$ is circular, the Riemann tensor associated with $g_{\mu\nu}^{(0)}$ is invertible (see the Supplemental Material \footnote{See Supplemental Material, which includes Refs. \cite{Carter:1969zz}, for a detailed discussion on tensor invertibility and its preservation through common operations.}\setcounter{aux_footnote}{\value{footnote}} for a proof).
Moreover, we show in the Supplemental Material \footnotemark[\value{aux_footnote}] that any tensor constructed from stationary, axisymmetric, and invertible tensors and their covariant derivatives associated with $g_{\mu\nu}^{(0)}$ is also itself invertible. Together with the assumption that the scalar fields $\varphi$ are stationary and axisymmetric, we conclude that $M_{\mu\nu}$ evaluated at zeroth order in $\alpha$ is invertible, and hence
\begin{eqnarray}\label{eq:M0}
\tensor{M}{_\mu^\nu} \zeta^\mu \eta_{\nu} = 0 + {\cal O}(\alpha).
\end{eqnarray}
Substituting Eq.~\eqref{eq:M0} into Eq.~\eqref{eq:M}, we find $\tensor{R}{_\mu^\nu} \zeta^\mu \,\tensor{\eta}{_\nu} $ vanishes to first order in $\alpha$, and therefore, $g_{\mu\nu}^{(1)}$ is circular.

Finally, we assume the solution is circular to the $n$th order in $\alpha$ and show that the solution to the $(n+1)$th order is circular. The proof is similar to that in the second step. In this case, $\tensor{M}{_\mu^\nu} \zeta^\mu \eta_{\nu}$ can be evaluated to the $n$th order in $\alpha$ with $g_{\mu\nu}^{(n)}$ and $\varphi^{(n)}$. The circularity of the $n$th order solution implies that
\begin{eqnarray}
\tensor{M}{_\mu^\nu} \zeta^\mu \eta_{\nu} = 0 + {\cal O}(\alpha^{n+1})\,.
\end{eqnarray}
Substituting this into Eq.~\eqref{eq:M}, we find $\tensor{R}{_\mu^\nu} \zeta^\mu \,\tensor{\eta}{_\nu}$ vanishes to the $(n+1)$th order in $\alpha$, and hence the solution to the $(n+1)$th order is circular. By induction, we conclude that the solution is circular to all orders in $\alpha$.

\emph{Extension to generalized light fields.---}%
Our proof can be further extended 
%
%
to theories with more general light fields, as long as the light fields and their leading-order stress-energy tensor $T_{\mu\nu}$ are invertible.

For light scalar fields, ${\cal L}_0$ may also include higher dimension operators that are arbitrary functions of $\varphi$, $\nabla_\mu \varphi \nabla^\mu \varphi$ and $\Box \varphi$. In this case, the resulting leading order stress-energy tensor is
\begin{align}\label{eq:Text}
T_{\mu\nu}=
&-\frac{\partial {\cal L}_0}{\partial(\nabla_\lambda\varphi\nabla^\lambda\varphi)}\nabla_\mu\varphi\nabla_\nu\varphi 
+\nabla_{(\mu}\left(\frac{\partial{\cal L}_0}{\partial(\Box\varphi)}\right)\nabla_{\nu)}\varphi \notag\\
&+\frac12 g_{\mu\nu}\left\{{\cal L}_0-\nabla_\lambda\left[\frac{\partial{\cal L}_0}{\partial(\Box\varphi)}\nabla^\lambda\varphi\right]\right\},
\end{align}
where terms proportional to $\nabla_\mu\varphi\nabla_\nu\varphi$ or $g_{\mu\nu}$ are invertible for the same reasons discussed above. Moreover, $\partial{\cal L}_0/\partial(\Box\varphi)$ inherits the symmetries of $\varphi$, so its Lie derivatives along $\zeta^\mu$ vanish.
Thus, the second term on the right hand side of Eq.~\eqref{eq:Text}, and hence the aggregated $T_{\mu\nu}$, is invertible, indicating stationary, axisymmetric and asymptotically flat vacuum solutions in such more general scalar-tensor theories are also circular.

In addition to the light scalars as described above, our proof can also be generalized to gravitational theories that include light vectors, as long as the nonstandard kinetic terms and nonminimal couplings to gravity may be treated perturbatively.
In particular, our proof can be extended to include light vectors with the following restrictions: (i) ${\cal L}_0$ is totally constructed from the vector fields $V_\mu$ and their exterior derivatives $F_{\mu\nu}=2 \nabla_{[\mu} V_{\nu]}$, and (ii) the vector fields $V_\mu$, apart from being stationary and axisymmetric, are invertible. 
In this case, since the exterior derivative does not depend on the metric, the energy-momentum tensor associated with ${\cal L}_0$ is completely constructed from $V_\mu$ and $F_{\mu\nu}$. We show in the Supplemental Material \footnotemark[\value{aux_footnote}] that $F_{\mu\nu}$ inherits the vector field's invertibility without assuming circularity. Therefore, $T_{\mu\nu}$ is invertible, and any such vector-tensor theory admits a circular ansatz for stationary and axisymmetric vacuum solutions. In addition, any generalized Proca theory, as introduced in~\cite{Tasinato:2014eka,Heisenberg:2014rta,Allys:2015sht,Allys:2016jaq,deRham:2020yet}, would inherit the same properties so long as the higher-order Lagrangians introduced in these theories are treated perturbatively.


\emph{Discussions.---}%
Our main result is a proof that the spacetime of stationary, axisymmetric, and asymptotically flat rotating black holes in a broad class of gravitational EFTs is circular. We emphasize that in addition to the light fields we have considered, the theory may also include any heavy field of arbitrary spin and coupling to gravity, as long as the mass of these fields is larger than the curvature scale of the black holes.
Our result is of immediate importance to the ongoing effort of testing the strong-field regime of gravity through gravitational waves~\cite{Gair:2012nm,Yunes:2013dva,Berti:2018cxi,Berti:2018vdi,Barausse:2020rsu} and electromagnetic observations~\cite{Krawczynski:2018fnw,Psaltis:2018xkc} in which black holes play a central role~\cite{Yagi:2016jml}.
These tests generally require knowledge of a rotating black hole solution (within a certain EFT) from which observable consequences are then deduced and then ultimately compared to observations.
Here, we proved that circularity is shared among a broad class of solutions, justifying the use of this ansatz when searching for analytical and numerical solutions. 

What are the implications of our result for some specific theories? Consider, for instance, dynamical Chern-Simons gravity in which a scalar field couples to the Pontryagin density~\cite{Jackiw:2003pm,Alexander:2009tp}. This theory must be treated as an EFT to admit a well-posed initial value problem~\cite{Delsate:2014hba}, and, in fact, this theory
is captured within the assumption of our proof.
Rotating black hole solutions in this theory are known numerically~\cite{Stein:2014xba,Delsate:2018ome} and analytically~\cite{Campbell:1990ai,Yunes:2009hc,Konno:2009kg,Ayzenberg:2014aka,Konno:2014qua} in 
a perturbative expansion in the coupling strength $\alpha$ and black hole spin $a \ll 1$ to ${\cal O}(\alpha^2 a^5)$~\cite{Yunes:2009hc,Yagi:2012ya,Maselli:2017kic} and in the extremal limit~\cite{McNees:2015srl}.
Our results indicate that the spacetime of rotating black holes in this theory is circular, justifying the use of the ansatz [Eq.~\eqref{eq:metric}] in numerical calculations.
The same applies to scalar Gauss-Bonnet gravity with shift-symmetric and dilatonic couplings where rotating black hole spacetimes are known analytically~\cite{Yunes:2011we,Pani:2011gy,Ayzenberg:2014aka,Maselli:2015tta,Maselli:2015yva} and numerically~\cite{Kleihaus:2011tg,Pani:2009wy,Delgado:2020rev,Sullivan:2020zpf}, including the final state of black holes that results at late times after highly dynamical black hole formation~\cite{Benkel:2016rlz,Ripley:2019aqj,Ripley:2019irj}.
%
In fact, it applies to 
any EFT extension of GR, including any low-energy EFT of gravity that includes massive fields of arbitrary spins.

Our results agree with those of~\cite{Nakashi:2020phm}, which suggested the nonexistence of rotating noncircular black holes in dynamical Chern-Simons gravity and shift-symmetric scalar-Gauss-Bonnet gravity, by working perturbatively to ${\cal O}(\alpha^2 a^2)$. Our conclusions extend to all orders in these two parameters.
Moreover, our results also apply to nonvacuum solutions in generic gravitational theories of the type discussed in this Letter, as long as the matter fields in the GR solution are stationary, axisymmetric, and possess an invertible stress-energy tensor. That is, our conclusion holds for a gravitational theory minimally coupled to an ordinary matter source, such as a perfect fluid that satisfies the same symmetries as the metric (i.e.~stationarity and axisymmetry).

We stress that our results only apply to solutions that reduce to a GR solution in the limit when the perturbative parameter $\alpha$ goes to zero. 
In general, this does not have to be the case, as other branches of solutions may be entropically favored, as is the case with theories that exhibit spontaneous black hole scalarization~\cite{Doneva:2017bvd,Silva:2017uqg}.

The requirement that the fields are stationary and axisymmetric (and invertible if of spin-1) is a sufficient but not a necessary condition for the solution to be circular, and it is not necessarily required by the isometries of the spacetime. There are cases in which the extra fields can be time- and angle-dependent, yet this dependence does not manifest itself in the gravitational equations. For example, there are hairy, nonlinear black hole solutions and solitonic solutions that arise in GR coupled to complex and massive (scalar) fields~\cite{Herdeiro:2014goa,Herdeiro:2015waa,Herdeiro:2018daq,Herdeiro:2019oqp}, where the metric is circular while the fields have time- or angle-dependent phases. Other examples are the stealth black holes of~\cite{Charmousis:2019vnf}, in which the scalar field has a linear time dependence, although such black hole solutions usually suffer from a strong coupling problem \cite{Babichev:2018uiw,deRham:2019gha,Ogawa:2015pea}. 

Our results do imply that if a theory satisfies the conditions of our theorem, then \textit{all} black hole solutions must have a circular spacetime, but the converse is not necessarily true. 
Imagine one were to find a black hole solution in a modified theory (in which our theorem does not apply) by requiring \textit{a priori} that the spacetime be circular. The existence of this solution does not then mean that other noncircular solutions do not exist.
For example, black hole solutions have been found in Einstein-Yang-Mills theories with~\cite{Kleihaus:2003sh} and without~\cite{Kleihaus:2000kg} a dilaton field, and in Einstein-\ae ther theory in the slow-rotation approximation~\cite{Barausse:2013nwa,Barausse:2015frm} assuming \textit{a priori} that the spacetime must be circular. In both cases, however, our theorem does not apply because either the Yang-Mills vector gauge field is noninvertible after gauge fixing or the \ae ther field is noninvertible because of its timelike constraint. Thus, the existence of those solutions does not imply that other noncircular black hole solutions do not exist in these theories, which could be explored further. 

\emph{Acknowledgements.---}%
We thank Lydia~Bieri, Daniela~Doneva, David~Garfinkle, Leonardo~Gualtieri, Carlos~A.~R.~Herdeiro, and Jutta~Kunz for discussions.
Y.X., H.O.S. and N.Y.~acknowledge financial support through NSF Grant
Nos.~PHY-1759615 and PHY-1949838 and NASA ATP Grant Nos.~17-ATP17-0225, NNX16AB98G, and 80NSSC17M0041.
C.d.R. and J.Z.~acknowledge financial support provided by the European Union's Horizon 2020 Research Council grant 724659 MassiveCosmo ERC-2016-COG.
H.W.~acknowledges financial support provided by the NSF Grant No. OAC-2004879 and the Royal Society Research Grant No.~RGF\textbackslash R1\textbackslash 180073.
C.d.R. also acknowledges financial support provided by STFC grants ST/P000762/1 and ST/T000791/1, by the Royal Society through a Wolfson Research Merit Award, by the Simons Foundation award ID 555326 under the Simons Foundation's Origins of the Universe initiative, ``Cosmology Beyond Einstein's Theory,'' and by the Simons Investigator award 690508. \vspace{-0.2cm}

\bibliography{references}

\clearpage
\appendix

\section*{Supplemental Material}

In this Supplemental Material we show some mathematical details that help identify the invertibility of tensors, in particular $T_{\mu\nu}$ and $M_{\mu\nu}$, as presented in the main text. Our discussion here follows Carter's work in \cite{Carter:1969zz}. As in that work, we derive results in $n$ dimensions with a group of isometries whose surfaces of transitivity are $p$-dimensional. In addition, we assume that the group is Abelian. To apply the results here to the proof in our main text, one sets $n=4$ and lets $\xi^\mu$ and $\chi^\mu$ be the generators of the group. Note that the orthogonal transitivity of this group in the 4-dimensional spacetime is referred to as ``circularity'' in the main text.

We begin with Carter's definition of tensor invertibility. 
Let $\svctv{\zeta}{i}{\mu}\,(i=1,\cdots,p)$ be a set of independent vectors spanning a $p$-dimensional surface element at a point $P$, and let $\svcov{\eta}{j}{\nu}\,(i=p+1,\cdots,n)$ be a set of independent vectors spanning the orthogonal $(n-p)$ element at $P$. A tensor $\tensor{T}{_{\mu_1\cdots\mu_r}^{\nu_1\cdots\nu_s}}$ is said to be invertible in the $p$ element at $P$ if all scalars contracted as
\begin{align}
    \tensor{T}{_{\mu_1\cdots\mu_r}^{\nu_1\cdots\nu_s}}\,
    \svctv{\zeta}{i_1}{\mu_1}\cdots\svctv{\zeta}{i_r}{\mu_r}\,
    \svcov{\eta}{j_1}{\nu_1}\cdots\svcov{\eta}{j_s}{\nu_s} \label{eq:cscalar}
\end{align}
are invariant when $\svctv{\zeta}{i}{\mu}\rightarrow-\svctv{\zeta}{i}{\mu}$ and $\svcov{\eta}{j}{\nu}\rightarrow\svcov{\eta}{j}{\nu}$ for all $i$, $j$. For convenience let us call all scalars with the form \eqref{eq:cscalar} \emph{Carter scalars}. The statement that a tensor is invertible in an element is equivalent to the statement that its Carter scalars vanish whenever $r$ is odd. Finally, a tensor is said to be invertible in a group if it is invertible to the surfaces of transitivity of the group. In our main text, the word ``invertible'' is short for ``invertible in the group generated by $\xi^\mu$ and $\chi^\mu$.''

Remarkably, Carter also mentioned another type of invertibility, namely skew invertibility, in which case the Carter scalars change sign when $\svctv{\zeta}{i}{\mu}\rightarrow-\svctv{\zeta}{i}{\mu}$ and $\svcov{\eta}{j}{\nu}\rightarrow\svcov{\eta}{j}{\nu}$, or equivalently they vanish whenever $r$ is even. Skew invertibility is not considered in our main text. However, it will be considered here, because it brings completeness and requires little extra work.

As suggested by their definition, invertibility and skew invertibility describes how a tensor responds to the simultaneous inversion of all Killing vectors -- an invertible tensor is invariant, a skew invertible tensor changes its sign, and a generic tensor has its invertible part and skew invertible part each transform accordingly. 
It sounds reasonable to guess that two tensors each with a certain type of invertibility (invertible or skew invertible) should combine into a third tensor with a certain type of invertibility through common tensor operations. In particular, the sum of invertible tensors should still be invertible, and the sum of skew invertible tensors should still be skew invertible. We will show below that this guess is true in a very broad sense. 
The preservation of invertibililty through operations is important, as it allows one to identify the invertibility of a tensor by looking at how it is constructed. 
In the main text, the invertibility of the Ricci tensor is related to the invertibility of $T_{\mu\nu}$ and $M_{\mu\nu}$ through the modified Einstein equation, and the invertibility of $T_{\mu\nu}$ and $M_{\mu\nu}$ is then related to the invertibility of their building blocks such as the scalar fields, the curvature tensor, and their derivatives. 

The rest of this Supplemental Material is formulated as lemmas and proofs, and they proceed as follows. First, we identify operations that preserve invertibility, starting with elementary operations including linear combination, trace, tensor product, and contraction. 
Then, we continue to investigate the action of derivatives. We show that a generic covariant derivative preserves invertibiltiy when the group is orthogonally transitive. But for covariant derivatives of scalar fields and exterior derivatives of vector fields, invertibility is preserved regardless of orthogonal transitivity. 
Finally, we explore the invertibility of tensors that are common ingredients in various theories of gravity. We show that generic scalars and the metric tensor are manifestly invertible, the Levi-Civita tensor is invertible or skew invertible depending on the dimension, and the Riemann curvature tensor is invertible when the group is orthogonally transitive.

\subsection{Preservation of Invertibility Through Elementry Operations}
\begin{lemma}\label{lm:inv_sum}
Let $P$ be a point on an $n$-dimensional manifold with a pseudo-Riemannian metric, such that there is a $p$-dimensional nonnull surface element at $P$. 
Let $T$ and $S$ be tensors of the same rank and with the same type of invertibility in the $p$ element at $P$. Then any linear combination of $T$ and $S$, $(aT+bS)$ where $a$ and $b$ are numbers, has the same type of invertibility as $T$ and $S$.
\end{lemma}

\begin{proof}
Let $\svctv{\zeta}{i}{\mu}\,(i=1,\cdots,p)$ be a set of independent vectors spanning the $p$-dimensional surface element at $P$, and let $\svctv{\eta}{i}{\mu}\,(i=p+1,\cdots,n)$ be a set of independent vectors spanning the orthogonal $(n-p)$ element at $P$. 

Consider the Carter scalar of $(aT+bS)$ at $P$:
\begin{align}
    &\tensor{(aT+bS)}{_{\mu_1\cdots\mu_r}_{\nu_1\cdots\nu_s}}
    \svctv{\zeta}{i_1}{\mu_1}\cdots\svctv{\zeta}{i_r}{\mu_r}
    \svctv{\eta}{j_1}{\nu_1}\cdots\svctv{\eta}{j_s}{\nu_s} \notag\\
    =&a\tensor{T}{_{\mu_1\cdots\mu_r}_{\nu_1\cdots\nu_s}}
    \svctv{\zeta}{i_1}{\mu_1}\cdots\svctv{\zeta}{i_r}{\mu_r}
    \svctv{\eta}{j_1}{\nu_1}\cdots\svctv{\eta}{j_s}{\nu_s} \notag\\
    &+b\tensor{S}{_{\mu_1\cdots\mu_r}_{\nu_1\cdots\nu_s}}
    \svctv{\zeta}{i_1}{\mu_1}\cdots\svctv{\zeta}{i_r}{\mu_r}
    \svctv{\eta}{j_1}{\nu_1}\cdots\svctv{\eta}{j_s}{\nu_s}. \label{eq:lm_sum}
\end{align}
If both $T$ and $S$ are invertible, then both terms on the right-hand side of (\ref{eq:lm_sum}) vanish whenever $r$ is odd, indicating that $(aT+bS)$ is invertible. If both $T$ and $S$ are skew invertible, then both terms on the right-hand side of (\ref{eq:lm_sum}) vanish whenever $r$ is even, indicating that $(aT+bS)$ is skew invertible. 
\end{proof}

\begin{lemma}\label{lm:inv_tr}
Let $P$ be a point on an $n$-dimensional manifold with a pseudo-Riemannian metric, such that there is a $p$-dimensional nonnull surface element at $P$. Let $T$ be a tensor which is invertible or skew invertible in the $p$ element at $P$. Then any trace of $T$, $\mathrm{tr}(T)$, has the same type of invertibility as $T$.
\end{lemma}

\begin{proof}
Let $\svctv{\zeta}{i}{\mu}\,(i=1,\cdots,p)$ be a set of independent vectors spanning the $p$-dimensional surface element at $P$, and let $\svctv{\eta}{i}{\mu}\,(i=p+1,\cdots,n)$ be a set of independent vectors spanning the orthogonal $(n-p)$ element at $P$. Given that the $p$ element is nonnull, the metric at $P$ takes the following form:
\begin{align}
    g^{\mu\nu}=\sum_{i,j=1}^p\! \tensor*[^{(ij)}]{A}{} \svctv{\zeta}{i}{\mu}\svctv{\zeta}{j}{\nu} + \sum_{i,j=p+1}^n\!\!\! \tensor*[^{(ij)}]{B}{} \svctv{\eta}{i}{\mu}\svctv{\eta}{j}{\nu},
\end{align}
where $\tensor*[^{(ij)}]{A}{}$ and $\tensor*[^{(ij)}]{B}{}$ are numbers. 

Consider the Carter scalar of $\mathrm{tr}(T)$ at $P$:
\begin{align}
    &\tensor{\mathrm{tr}(T)}{_{\mu_1\cdots\mu_r\nu_1\cdots\nu_s}} \svctv{\zeta}{i_1}{\mu_1}\cdots\svctv{\zeta}{i_r}{\mu_r}
    \svctv{\eta}{j_1}{\nu_1}\cdots\svctv{\eta}{j_s}{\nu_s} \notag\\
    =&g^{\rho\sigma}\tensor{T}{_{\rho\sigma\mu_1\cdots\mu_r\nu_1\cdots\nu_s}} \svctv{\zeta}{i_1}{\mu_1}\cdots\svctv{\zeta}{i_r}{\mu_r}
    \svctv{\eta}{j_1}{\nu_1}\cdots\svctv{\eta}{j_s}{\nu_s} \notag\\
    =&\sum_{m,n=1}^p\! \tensor*[^{(mn)}]{A}{}
    \svctv{\zeta}{m}{\rho}\svctv{\zeta}{n}{\sigma}
    \tensor{T}{_{\rho\sigma\mu_1\cdots\mu_r\nu_1\cdots\nu_s}} \notag\\
    &\qquad\quad\times\svctv{\zeta}{i_1}{\mu_1}\cdots\svctv{\zeta}{i_r}{\mu_r}
    \svctv{\eta}{j_1}{\nu_1}\cdots\svctv{\eta}{j_s}{\nu_s} \notag\\
    &+\sum_{m,n=p+1}^{n}\!\!\! \tensor*[^{(mn)}]{B}{}
    \svctv{\eta}{m}{\rho}\svctv{\eta}{n}{\sigma}
    \tensor{T}{_{\rho\sigma\mu_1\cdots\mu_r\nu_1\cdots\nu_s}} \notag\\
    &\qquad\quad\times\svctv{\zeta}{i_1}{\mu_1}\cdots\svctv{\zeta}{i_r}{\mu_r}
    \svctv{\eta}{j_1}{\nu_1}\cdots\svctv{\eta}{j_s}{\nu_s}. \label{eq:lm_tr}
\end{align}
The first term sums over Carter scalars of $T$ with $(r+2)$ $\zeta$'s, and the second term sums over Carter scalars of $T$ with $r$ $\zeta$'s. If $T$ is invertible, then both sums vanish whenever $r$ is odd, indicating that $\mathrm{tr}(T)$ is invertible. Similarly one can also prove that $\mathrm{tr}(T)$ is skew invertible if $T$ is skew invertible. 
\end{proof}

\begin{lemma}\label{lm:inv_prod}
Let $P$ be a point on an $n$-dimensional manifold with a pseudo-Riemannian metric, such that there is a $p$-dimensional nonnull surface element at $P$. 
Let $T$ and $S$ be tensors which are invertible or skew invertible in the $p$ element at $P$. Then the tensor product of $T$ and $S$, $T\otimes S$, is invertible $T$ and $S$ are both invertible or both skew invertible, and is skew invertible $P$ if one of $T$ and $S$ is invertible and the other is skew invertible.
\end{lemma}

\begin{proof}
Let $\svctv{\zeta}{i}{\mu}\,(i=1,\cdots,p)$ be a set of independent vectors spanning the $p$-dimensional surface element at $P$, and let $\svctv{\eta}{i}{\mu}\,(i=p+1,\cdots,n)$ be a set of independent vectors spanning the orthogonal $(n-p)$ element at $P$. 

Consider the Carter scalar of $T\otimes S$ at $P$:
\begin{align}
    &\tensor{(T\otimes S)}{_{\mu_1\cdots\mu_r\rho_1\cdots\rho_u \nu_1\cdots\nu_s\sigma_1\cdots\sigma_v}}  \notag\\
    &\times \svctv{\zeta}{i_1}{\mu_1}\cdots\svctv{\zeta}{i_r}{\mu_r}
    \svctv{\zeta}{j_1}{\nu_1}\cdots\svctv{\zeta}{j_s}{\mu_s} \notag\\
    &\times \svctv{\eta}{k_1}{\rho_1}\cdots\svctv{\eta}{k_u}{\rho_u}
    \svctv{\eta}{l_1}{\sigma_1}\cdots\svctv{\eta}{l_v}{\sigma_v} \notag\\
    =&\tensor{T}{_{\mu_1\cdots\mu_r\rho_1\cdots\rho_u}}
    \svctv{\zeta}{i_1}{\mu_1}\cdots\svctv{\zeta}{i_r}{\mu_r}
    \svctv{\eta}{k_1}{\rho_1}\cdots\svctv{\eta}{k_u}{\rho_u} \notag\\
    &\times\tensor{S}{_{\nu_1\cdots\nu_s\sigma_1\cdots\sigma_v}}
    \svctv{\zeta}{j_1}{\nu_1}\cdots\svctv{\zeta}{j_s}{\mu_s}
    \svctv{\eta}{l_1}{\sigma_1}\cdots\svctv{\eta}{l_v}{\sigma_v}. 
    \label{eq:lm_prod}
\end{align}
If $T$ and $S$ are both invertible, consider that $(r+s)$ is odd, then either $r$ or $s$ is odd, making either the $T$ part or the $S$ part on the right-hand side vanish. So in this case $T\otimes S$ is invertible. Similarly one can also show that $T\otimes S$ is invertible if $T$ and $S$ are both skew invertible. 

If $T$ is invertible and $S$ is skew invertible, consider that $(r+s)$ is even, then $r$ and $s$ are either both odd or both even, making either the $T$ part or the $S$ part on the right-hand side vanish. So in this case $T\otimes S$ is skew invertible. Similarly one can also show that $T\otimes S$ is skew invertible if $T$ is skew invertible and $S$ is invertible. 
\end{proof}

\begin{lemma}\label{lm:inv_contr}
Let $P$ be a point on an $n$-dimensional manifold with a pseudo-Riemannian metric, such that there is a $p$-dimensional nonnull surface element at $P$. 
Let $T$ and $S$ be tensors which are invertible or skew invertible in the $p$ element at $P$. Then any contraction of $T$ and $S$, $T\cdot S$, is invertible if $T$ and $S$ are both invertible or both skew invertible, and is skew invertible if one of $T$ and $S$ is invertible and the other is skew invertible.
\end{lemma}

\begin{proof}
Since any contraction can be treated as a tensor product followed by a series of tracing, we may apply Lemma \ref{lm:inv_prod} and Lemma \ref{lm:inv_tr}. In particular, Lemma \ref{lm:inv_prod} suggests that $T\otimes S$ is invertible if $T$ and $S$ are both invertible or both skew invertible, and is skew invertible if one of $T$ and $S$ is invertible and the other is skew invertible. Then Lemma \ref{lm:inv_tr} suggests that $T\cdot S=\mathrm{tr}(T\otimes S)$ has the same type of invertiblity of $T\otimes S$.
\end{proof}

\subsection{Preservation of Invertibility Through Derivatives}
\begin{lemma} \label{lm:inv_chris}
Let $\mathcal{U}$ be an open subregion of an $n$-dimensional manifold with a pseudo-Riemannian metric, such that there is an Abelian $p$-parameter isometry group which is orthogonally transitive in $\mathcal{U}$.
Consequently, in $\mathcal{U}$ there exist $p$ independent Killing vectors generating the surfaces of transitivity, and $(n-p)$ independent vectors orthogonal to the surfaces of transitivity, such that they all commute. 
Let $\svctv{\zeta}{i}{\mu}\,(i=1,\cdots,p)$ be the Killing vectors, and $\svctv{\eta}{i}{\mu}\,(i=p+1,\cdots,n)$ be the orthogonal vectors. Then $\forall i$, $\nabla\svctv{\zeta}{i}{}$ is skew invertible, and $\nabla\svctv{\eta}{i}{}$ is invertible, in the group everywhere in $\mathcal{U}$. 
\end{lemma}

\begin{proof}
This is equivalent to proving that all following Carter scalars are zero:
\begin{align*} 
&\svctv{\zeta}{i}{\mu}\svctv{\zeta}{j}{\nu}\nabla_\mu(\svcov{\zeta}{k}{\nu}),
\quad\svctv{\eta}{i}{\mu}\svctv{\eta}{j}{\nu}\nabla_\mu(\svcov{\zeta}{k}{\nu}),\\
&\svctv{\zeta}{i}{\mu}\svctv{\eta}{j}{\nu}\nabla_\mu(\svcov{\eta}{k}{\nu}),
\quad\svctv{\eta}{i}{\mu}\svctv{\zeta}{j}{\nu}\nabla_\mu(\svcov{\eta}{k}{\nu}).
\end{align*}
In the first scalar, $(i)(j)$ are antisymmetric due to the Killing equation of $\svctv{\zeta}{k}{}$, and $(i)(k)$ are symmetric because $\svctv{\zeta}{i}{}$ and $\svctv{\zeta}{k}{}$ commute. With these in mind, we have
\begin{align}
    &\svctv{\zeta}{i}{\mu}\svctv{\zeta}{j}{\nu}\nabla_\mu(\svcov{\zeta}{k}{\nu})
    =-\svctv{\zeta}{j}{\mu}\svctv{\zeta}{i}{\nu}\nabla_\mu(\svcov{\zeta}{k}{\nu})\notag\\
    =&-\svctv{\zeta}{k}{\mu}\svctv{\zeta}{i}{\nu}\nabla_\mu(\svcov{\zeta}{j}{\nu})
    =\svctv{\zeta}{i}{\mu}\svctv{\zeta}{k}{\nu}\nabla_\mu(\svcov{\zeta}{j}{\nu})\notag\\
    =&\svctv{\zeta}{j}{\mu}\svctv{\zeta}{k}{\nu}\nabla_\mu(\svcov{\zeta}{i}{\nu})
    =-\svctv{\zeta}{k}{\mu}\svctv{\zeta}{j}{\nu}\nabla_\mu(\svcov{\zeta}{i}{\nu})\notag\\
    =&-\svctv{\zeta}{i}{\mu}\svctv{\zeta}{j}{\nu}\nabla_\mu(\svcov{\zeta}{k}{\nu})=0.
\end{align}
For the second scalar, again using the Killing equation, together with the orthogonality and commutation,
\begin{align}
    &\svctv{\eta}{i}{\mu}\svctv{\eta}{j}{\nu}\nabla_\mu(\svcov{\zeta}{k}{\nu})
    =-\svctv{\eta}{j}{\mu}\svctv{\eta}{i}{\nu}\nabla_\mu(\svcov{\zeta}{k}{\nu})\notag\\
    =&\svctv{\eta}{j}{\mu}\svctv{\zeta}{k}{\nu}\nabla_\mu(\svcov{\eta}{i}{\nu})
    =\svctv{\eta}{i}{\mu}\svctv{\zeta}{k}{\nu}\nabla_\mu(\svcov{\eta}{j}{\nu})\notag\\
    =&-\svctv{\eta}{i}{\mu}\svctv{\eta}{j}{\nu}\nabla_\mu(\svcov{\zeta}{k}{\nu})=0.
\end{align}
This result can then be applied to the last two scalars:
\begin{align}
    &\svctv{\zeta}{i}{\mu}\svctv{\eta}{j}{\nu}\nabla_\mu(\svcov{\eta}{k}{\nu})
    =\svctv{\eta}{k}{\mu}\svctv{\eta}{j}{\nu}\nabla_\mu(\svcov{\zeta}{i}{\nu})=0, \\
    &\svctv{\eta}{i}{\mu}\svctv{\zeta}{j}{\nu}\nabla_\mu(\svcov{\eta}{k}{\nu})
    =-\svctv{\eta}{i}{\mu}\svctv{\eta}{k}{\nu}\nabla_\mu(\svcov{\zeta}{j}{\nu})=0.
\end{align}
Thus $\nabla\svctv{\zeta}{k}{}$ is skew invertible and $\nabla\svctv{\eta}{k}{}$ is invertible.
\end{proof}

\begin{lemma} \label{lm:inv_d}
Let $\mathcal{U}$ be an open subregion of an $n$-dimensional manifold with a pseudo-Riemannian metric, such that there is an Abelian $p$-parameter isometry group which is orthogonally transitive in $\mathcal{U}$.
Let $T$ be a tensor which is invertible or skew invertible in the group everywhere in $\mathcal{U}$, and is Lie transported by the generators of the group. Then the covariant derivative of $T$, $\nabla T$, has the same type of invertibility as $T$.
\end{lemma}

\begin{proof}
Since the group is Abelian and orthogonally transitive, there exist independent Killing vectors $\svctv{\zeta}{i}{\mu}\,(i=1,\cdots,p)$ generating the surfaces of transitivity, and independent vectors $\svctv{\eta}{i}{\mu}\,(i=p+1,\cdots,n)$ orthogonal to the surfaces of transitivity, such that they all commute. 

Consider the Carter scalars of $\nabla T$. Typically there are two types of these scalars depending on whether the $\nabla$ is contracted with a $\zeta$ or $\eta$. In the $\zeta\cdot\nabla$ case, let $\svctv{\zeta}{i_1}{}$ be the one contracting with $\nabla$, and the condition that $T$ is Lie transported by the group generators tells us that
\begin{align}
    &\svctv{\zeta}{i_1}{\mu_1}\nabla_{\mu_1}\tensor{T}{_{\mu_2\cdots\mu_{t}}}\notag\\
    &\quad=-\sum_{a=2}^{t}\nabla_{\mu_a}(\svctv{\zeta}{i_1}{\mu_1})\, \tensor{T}{_{\mu_2\cdots\mu_{a-1}\mu_1\mu_{a+1}\cdots\mu_{t}}}.
\end{align}
Thus the Carter scalar is
\begin{align}
    &\nabla_{\mu_1}\tensor{T}{_{\mu_2\cdots\mu_{r+s}}}\notag\\
    &\quad\times\svctv{\zeta}{i_1}{\mu_1}\cdots\svctv{\zeta}{i_r}{\mu_{r}}
    \svctv{\eta}{i_{r+1}}{\mu_{r+1}}\cdots\svctv{\zeta}{i_{r+s}}{\mu_{r+s}} \notag\\
    =&-\sum_{a=2}^{r+s} 
    \nabla_{\mu_a}(\svctv{\zeta}{i_1}{\mu_1})\,\tensor{T}{_{\mu_2\cdots\mu_{a-1}\mu_1\mu_{a+1}\cdots\mu_{r+s}}} \notag\\
    &\quad\times\svctv{\zeta}{i_2}{\mu_2}\cdots\svctv{\zeta}{i_{r}}{\mu_{r}} \svctv{\eta}{i_{r+1}}{\mu_{r+1}}\cdots\svctv{\eta}{i_{r+s}}{\mu_{r+s}}. \label{eq:lm_dzeta}
\end{align}
The right-hand side is a summation of Carter scalars of $(\nabla\zeta)\cdot T$ with $(r-1)$ directions of transitivity. According to Lemma \ref{lm:inv_chris}, $\nabla\zeta$ is skew invertible. Then according to Lemma \ref{lm:inv_contr}, if $T$ is invertible, $(\nabla\zeta)\cdot T$ is skew invertible, and the summation becomes zero whenever $r$ is odd. If $T$ is skew invertible, $(\nabla\zeta)\cdot T$ is invertible, and the summation becomes zero whenever $r$ is even.

In the $\eta\cdot\nabla$ case, let $\svctv{\eta}{j_1}{}$ be the one contracting with $\nabla$. The Carter scalar is
\begin{align}
    &\nabla_{\nu_1} \tensor{T}{_{\mu_1\cdots\mu_r\nu_2\cdots\nu_{s}}}
    \svctv{\zeta}{i_1}{\mu_1}\cdots\svctv{\zeta}{i_r}{\mu_r}
    \svctv{\eta}{j_1}{\nu_1}\cdots\svctv{\zeta}{j_{s}}{\nu_{s}} \notag\\
    =&\svctv{\eta}{j_1}{\nu_1}\nabla_{\nu_1} \Big( \tensor{T}{_{\mu_1\cdots\mu_r\nu_2\cdots\nu_{s}}} \notag\\
    &\quad\times\svctv{\zeta}{i_1}{\mu_1}\cdots\svctv{\zeta}{i_r}{\mu_r}
    \svctv{\eta}{j_2}{\nu_2}\cdots\svctv{\zeta}{j_{s}}{\nu_{s}} \Big) \notag\\
    &-\sum_{a=1}^{r}
    \nabla_{\nu_1}(\svctv{\zeta}{i_a}{\mu_a})\,
    \tensor{T}{_{\mu_1\cdots\mu_{r}\nu_2\cdots\nu_{s}}} \svctv{\eta}{j_{1}}{\nu_{1}}\cdots\svctv{\eta}{j_{s}}{\nu_{s}} \notag\\
    &\quad\times \svctv{\zeta}{i_1}{\mu_1}\cdots\svctv{\zeta}{i_{a-1}}{\mu_{a-1}}\svctv{\zeta}{i_{a+1}}{\mu_{a+1}}\cdots\svctv{\zeta}{i_r}{\mu_r} \notag\\
    &-\sum_{a=2}^{s} 
    \nabla_{\nu_1}(\svctv{\eta}{j_a}{\nu_a})\, \tensor{T}{_{\mu_1\cdots\mu_r\nu_2\cdots\nu_{s}}}
    \svctv{\zeta}{i_1}{\mu_1}\cdots\svctv{\zeta}{i_r}{\mu_r} \notag\\
    &\quad\times \svctv{\eta}{j_1}{\nu_1}\cdots\svctv{\eta}{j_{a-1}}{\nu_{a-1}}\svctv{\eta}{j_{a+1}}{\nu_{a+1}}\cdots\svctv{\eta}{j_{s}}{\nu_{s}}. \label{eq:lm_deta}
\end{align}
The first term is a derivative of a Carter scalar created upon $T$ with $r$ directions of transitivity, the second term sums over Carter scalars of $(\nabla\zeta)\cdot T$ with $(r-1)$ directions of transitivity, and the third term sums over Carter scalars of $(\nabla\eta)\cdot T$ with $r$ directions of transitivity. We again apply Lemma \ref{lm:inv_chris} and Lemma \ref{lm:inv_contr}. If $T$ is invertible, then all three terms vanish whenever $r$ is odd. If $T$ is skew invertible, then all three terms vanish whenever $r$ is even.

To summarize, if $T$ is invertible, then any Carter scalar of $\nabla T$ with an odd $r$ vanishes. If $T$ is skew invertible, then any the Carter scalar of $\nabla T$ with an even $r$ vanishes. In other words, $\nabla T$ is invertible if $T$ is invertible, and is skew invertible if $T$ is skew invertible.
\end{proof}

\begin{lemma}\label{lm:inv_dphi}
Let $\mathcal{U}$ be an open subregion of an $n$-dimensional manifold with a pseudo-Riemannian metric, such that there is an Abelian $p$-parameter isometry group in $\mathcal{U}$. Let $\varphi$ be a scalar defined on $\mathcal{U}$ which is Lie transported by the generators of the group. Then $\nabla\varphi$ is invertible in the group everywhere in $\mathcal{U}$.
\end{lemma}

\begin{proof}
Let $\svctv{\zeta}{i}{\mu}\,(i=1,\cdots,p)$ be independent Killing vectors generating the surfaces of transitivity. The condition that $\varphi$ is Lie transported by the group generators suggests that
\begin{align}
\svctv{\zeta}{i}{\mu}\nabla_\mu\varphi=0.
\end{align}
Note that the left-hand side also represents all Carter scalars of $\nabla\varphi$ with an odd number of $\zeta$'s. Therefore this condition indicates that $\nabla\varphi$ is invertible. 
\end{proof}

\begin{lemma}\label{lm:inv_dA}
Let $\mathcal{U}$ be an open subregion of an $n$-dimensional manifold with a pseudo-Riemannian metric, such that there is an Abelian $p$-parameter isometry group in $\mathcal{U}$.
Let $A$ be a 1-form defined on $\mathcal{U}$ which is both Lie transported by the generators of the group and invertible in the group. Then the exterior derivative of $A$, $({\rm d}A)_{\mu\nu}\equiv\nabla_\mu A_\nu-\nabla_\nu A_\mu$, is invertible in the group everywhere in $\mathcal{U}$.
\end{lemma}

\begin{proof}
Let $\svctv{\zeta}{i}{\mu}\,(i=1,\cdots,p)$ be independent Killing vectors generating the surfaces of transitivity, and $\svctv{\eta}{i}{\mu}\,(i=p+1,\cdots,n)$ be independent vectors orthogonal to the surfaces of transitivity. The condition that $A$ is Lie transported by the group generators suggests that
\begin{align}
\svctv{\zeta}{i}{\mu}\nabla_\mu A_\nu+A_\mu\nabla_\nu\svctv{\zeta}{i}{\mu}=0.
\end{align}
Due to the antisymmetry of ${\rm d}A$, we only needs to show that the following Carter scalar vanishes:
\begin{align}
\svctv{\zeta}{i}{\mu}\svctv{\eta}{j}{\nu}({\rm d}A)_{\mu\nu}
=&\svctv{\zeta}{i}{\mu}\svctv{\eta}{j}{\nu}(\nabla_\mu A_\nu-\nabla_\nu A_\mu) \notag\\
=&-\svctv{\eta}{j}{\nu}(A_\mu\nabla_\nu\svctv{\zeta}{i}{\mu}+\svctv{\zeta}{i}{\mu}\nabla_\nu A_\mu) \notag\\
=&-\nabla_\nu(A_\mu\svctv{\zeta}{i}{\mu})\svctv{\eta}{j}{\nu}.
\end{align}
Since $A$ is invertible, $A_\mu\svctv{\zeta}{i}{\mu}$ vanishes, and this Carter scalar also vanishes. Therefore ${\rm d}A$ is invertible.
\end{proof}

\subsection{Invertibility of Common Tensors}
\begin{lemma}\label{lm:inv_phi}
Let $P$ be a point on an $n$-dimensional manifold with a pseudo-Riemannian metric, such that there is a $p$-dimensional nonnull surface element at $P$. Then any scalar $\varphi$ is invertible in the $p$ element at $P$. 
\end{lemma}

\begin{proof}
The only Carter scalar that can be created upon $\varphi$ is $\varphi$ itself, which is invariant when all directions on the $p$ element are simultaneously inverted. Therefore $\varphi$ is invertible.
\end{proof}

\begin{lemma}\label{lm:inv_g}
Let $P$ be a point on an $n$-dimensional manifold with a pseudo-Riemannian metric, such that there is a $p$-dimensional nonnull surface element at $P$. Then the metric $g_{\mu\nu}$ is invertible in the $p$ element at $P$. 
\end{lemma}

\begin{proof}
Let $\svctv{\zeta}{i}{\mu}\,(i=1,\cdots,p)$ be a set of independent vectors spanning the $p$-dimensional surface element at $P$, and let $\svctv{\eta}{i}{\mu}\,(i=p+1,\cdots,n)$ be a set of independent vectors spanning the orthogonal $(n-p)$ element at $P$. By orthogonality we have
\begin{align}
    g_{\mu\nu}\svctv{\zeta}{i}{\mu}\svctv{\eta}{j}{\nu}=0.
\end{align}
Note that the left-hand side also represents all Carter scalars of $g_{\mu\nu}$ with an odd number of $\zeta$'s. Therefore this orthogonality indicates that $g_{\mu\nu}$ is invertible. 
\end{proof}

\begin{lemma}\label{lm:inv_e}
Let $P$ be a point on an $n$-dimensional manifold with a pseudo-Riemannian metric, such that there is a $p$-dimensional nonnull surface element at $P$. Then the Levi-Civita tensor $\epsilon_{\mu_1\cdots\mu_n}$ is invertible in the $p$ element at $P$ if $p$ is even, and is skew invertible in the $p$ element at $P$ if $p$ is odd. 
\end{lemma}

\begin{proof}
Let $\svctv{\zeta}{i}{\mu}\,(i=1,\cdots,p)$ be a set of independent vectors spanning the $p$-dimensional surface element at $P$, and let $\svctv{\eta}{i}{\mu}\,(i=p+1,\cdots,n)$ be a set of independent vectors spanning the orthogonal $(n-p)$ element at $P$. 

Consider the Carter scalar of $\epsilon_{\mu_1\cdots\mu_n}$ at $P$:
\begin{align}
    &\tensor{\epsilon}{_{\mu_1\cdots\mu_r}_{\nu_1\cdots\nu_{n-r}}}
    \svctv{\zeta}{i_1}{\mu_1}\cdots\svctv{\zeta}{i_r}{\mu_r}
    \svctv{\eta}{j_1}{\nu_1}\cdots\svctv{\eta}{j_{n-r}}{\nu_{n-r}}.
\end{align}
Given the antisymmetrization, this scalar vanishes $\forall r\neq p$. If $p$ is even, then this scalar vanishes whenever $r$ is odd, indicating that $\epsilon_{\mu_1\cdots\mu_n}$ is invertible. If $p$ is odd, then this scalar vanishes whenever $r$ is even, indicating that $\epsilon_{\mu_1\cdots\mu_n}$ is skew invertible. 
\end{proof}

\begin{lemma} \label{lm:inv_R}
Let $\mathcal{U}$ be an open subregion of an $n$-dimensional manifold with a pseudo-Riemannian metric, such that there is an Abelian $p$-parameter isometry group which is orthogonally transitive in $\mathcal{U}$. Then the Riemann tensor $R_{\rho\sigma\mu\nu}$ is invertible in the group everywhere in $\mathcal{U}$.
\end{lemma}

\begin{proof}
Since the group is Abelian and orthogonally transitive, there exist independent Killing vectors $\svctv{\zeta}{i}{\mu}\,(i=1,\cdots,p)$ generating the surfaces of transitivity, and independent vectors $\svctv{\eta}{i}{\mu}\,(i=p+1,\cdots,n)$ orthogonal to the surfaces of transitivity, such that they all commute. 

Due to the symmetry of the Riemann tensor, we only need to show that the following two Carter scalars vanish:
\begin{align*}
R_{\rho\sigma\mu\nu}\svctv{\eta}{k}{\rho}\svctv{\eta}{l}{\sigma}\svctv{\eta}{i}{\mu}\svctv{\zeta}{j}{\nu}, ~
R_{\rho\sigma\mu\nu}\svctv{\eta}{k}{\rho}\svctv{\zeta}{l}{\sigma}\svctv{\zeta}{i}{\mu}\svctv{\zeta}{j}{\nu}.
\end{align*}
Note that as a consequence of the Killing equation,
\begin{align}
R_{\rho\sigma\mu\nu}\svctv{\zeta}{j}{\nu}=\nabla_\mu\nabla_\sigma\svcov{\zeta}{j}{\rho}.
\end{align}
So the above Carter scalars can be treated as Carter scalars of $\nabla\nabla\svcov{\zeta}{j}{}$, i.e.
\begin{align}
&R_{\rho\sigma\mu\nu}\svctv{\eta}{k}{\rho}\svctv{\eta}{l}{\sigma}\svctv{\eta}{i}{\mu}\svctv{\zeta}{j}{\nu}\notag\\
&~\quad=\svctv{\eta}{k}{\rho}\svctv{\eta}{l}{\sigma}\svctv{\eta}{i}{\mu}\nabla_\mu\nabla_\sigma\svcov{\zeta}{j}{\rho},\\
&R_{\rho\sigma\mu\nu}\svctv{\eta}{k}{\rho}\svctv{\zeta}{l}{\sigma}\svctv{\zeta}{i}{\mu}\svctv{\zeta}{j}{\nu}\notag\\
&~\quad=\svctv{\eta}{k}{\rho}\svctv{\zeta}{l}{\sigma}\svctv{\zeta}{i}{\mu}\nabla_\mu\nabla_\sigma\svcov{\zeta}{j}{\rho}.
\end{align}
According to Lemma \ref{lm:inv_chris}, $\nabla\svcov{\zeta}{j}{}$ is skew invertible. Then according to Lemma \ref{lm:inv_d}, $\nabla\nabla\svcov{\zeta}{j}{}$ is also skew invertible. Therefore the above two Carter scalars both vanish, and $R_{\rho\sigma\mu\nu}$ is invertible. 
\end{proof}


\end{document}